\documentclass[sigconf]{acmart}

\usepackage{booktabs} % For formal tables

\usepackage{tikz}
\usetikzlibrary{shapes.geometric, arrows, positioning, backgrounds}

\usepackage{algpseudocode}
\usepackage{color, soul}	

\usepackage{enumitem}
\usepackage{balance}

\newcommand{\cl}[1]{\mathcal{#1}} % for caligraphic symbols
\newcommand{\NN}{\mathbb{N}}      % for Naturals
\newcommand{\EE}{\mathbb{E}}	  % for expectation
\DeclareMathOperator*{\argmax}{argmax}

\setlist[itemize,1]{leftmargin=\dimexpr 13pt}
\setlist[enumerate,1]{leftmargin=\dimexpr 13pt}

\renewcommand\footnotetextcopyrightpermission[1]{} % removes footnote with conference information in first column
\begin{document}

\fancyhead{}  
\title{Incentive-Compatible Diffusion}
\titlenote{This project has received funding from the European Research Council (ERC) under the European Union's Horizon 2020 research and innovation programme (grant agreement  n$^o  $ 740435).}

\author{Yakov Babichenko}
\affiliation{%
	\institution{Technion --- Israel Institute of Technology}
	\city{Haifa} 
	\country{Israel}}

\email{yakovbab@tx.technion.ac.il}

\author{Oren Dean}

\affiliation{%
	\institution{Technion --- Israel Institute of Technology}
	\city{Haifa} 
	\country{Israel}}
\email{orendean@campus.technion.ac.il}

\author{Moshe Tennenholtz}
\affiliation{%
	\institution{Technion --- Israel Institute of Technology}
	\city{Haifa} 
	\country{Israel}}
\email{moshet@ie.technion.ac.il}

\begin{abstract}
Our work bridges the literature on incentive-compatible mechanism design and the literature on diffusion algorithms.
We introduce the study of finding an incentive-compatible (strategy-proof) mechanism for selecting an influential vertex in a directed graph (e.g. Twitter's network). The goal is to devise a mechanism with a bounded ratio between the maximal influence and the influence of the selected user, and in which no user can improve its probability of being selected by following or unfollowing other users. We introduce the \emph{Two Path} mechanism which is based on the idea of selecting the vertex that is the first intersection of two independent random walks in the network. The Two Path mechanism is incentive compatible on directed acyclic graphs (DAGs), and has a finite approximation ratio on natural subfamilies of DAGs. Simulations indicate that this mechanism is suitable for practical uses.

\end{abstract}


\maketitle

\section{Introduction}
	
	The problem of selecting influential users in a network, has been extensively studied in the past decade (\cite{PPRank,Sun:2012:IIU:2310057.2310059,doi:10.1509/jmkr.47.4.643}). These works mainly focus on the problem of maximizing diffusion in a \emph{given} network. If we assume that users of the network are willing to be selected by the mechanism, it is desirable that the selection mechanism will not only have good diffusion performance, but also be \emph{incentive compatible}; namely, that a user could not affect his probability of being selected by strategic erasure (or formation) of links. For instance, suppose that Twitter would like to start a new epidemic trend. It would like to find a very influential user and give him some benefit in order to kickoff its campaign. If the users are aware of these intents, they might remove or add out-going links in order to get chosen. A similar problem would be to select an academic paper in a particular field and awarding it as `the most influential' in this field; measuring the influence according to quotes from other papers and taking into account both direct and propagated influence.
	
	The following simple example demonstrates the tendency between diffusion maximization and incentive compatibility.
	
	\begin{tikzpicture}[scale=1, line width=0.4mm]

		\node[label={[font=\normalsize, blue]},circle,draw, text=black](A) at (0,0){A};
		\node[label={[font=\normalsize, blue]},circle,draw, text=black](B) at (-20pt,-20pt){B};
		\node[label={[font=\normalsize, blue]},circle,draw, text=black](C) at (20pt,-20pt){C};
		\node[label={[font=\normalsize, blue]},circle,draw, text=black](D) at (0pt,-40pt){D};
		\node[label={[font=\normalsize, blue]},circle,draw, text=black](E) at (40pt,-40pt){E};

		\draw[->] (B) -- (A);
		\draw[->] (C) -- (A);
		\draw[->] (D) -- (C);
		\draw[->] (E) -- (C);
		
		\node[label={[font=\normalsize, blue]},circle,draw, text=black](A') at (120pt,0pt){A};
		\node[label={[font=\normalsize, blue]},circle,draw, text=black](B') at (100pt,-20pt){B};
		\node[label={[font=\normalsize, blue]},circle,draw, text=black](C') at (140pt,-20pt){C};
		\node[label={[font=\normalsize, blue]},circle,draw, text=black](D') at (120pt,-40pt){D};
		\node[label={[font=\normalsize, blue]},circle,draw, text=black](E') at (160pt,-40pt){E};

		\draw[->] (B') -- (A');
		\draw[->] (D') -- (C');
		\draw[->] (E') -- (C');
		
		\draw(10pt,-50pt) node[below] {$G_1$};
		\draw(130pt,-50pt) node[below] {$G_2$};
		
		\end{tikzpicture}
		
A directed edge from $u$ to $v$ indicates that $u$ follows (or cites) $v$. In any plausible diffusion model, the most influential vertex in network $G_1$ is $A$, while the most influential vertex in network $G_2$ is $C$. Suppose the mechanism simply selects the most influential vertex. Then user $C$, knowing that the network is $ G_1 $, might have an incentive not to follow user $A$ (although she is interested in user's $A$ content) in order to be selected by the mechanism.

As a first step towards understanding incentive compatible diffusion mechanisms, our main focus in this paper is on the problem of selecting a single user when the graph is acyclic. Note that in environments where users arrive in horological order (e.g., scientific papers or scale-free networks; see, for example, \cite{2000PhRvL..85.4633D,2002AdPhy..51.1079D}) acyclic graphs are indeed a good proxy for the actual environment.

One of the simplest ideas to develop mechanisms (hopefully incentive compatible) that select vertices with (hopefully) high diffusion is the following: choose a node $v_1$ (at random) and ask him to mention a single user $v_2$ that he follows. Informally speaking, by arguments in the spirit of the friendship paradox (\cite{Feld91, Lattanzi:2015:PRN:2684822.2685293}), we expect that $v_2$ will have better diffusion than $v_1$. We can then ask $v_2$ who is his friend, etcetera. Thus, proceeding along a path increases (in expectation) the influence of the observed user. Two obstacles arise with this idea:
\begin{itemize}
	\item When do we stop the process?
	
	\item What if a node reports that he does not follow anyone?\footnote{Note that we cannot stop and select him, as this will incentive every node to report that he follows no one.}.
\end{itemize}
The \emph{Two Path} mechanism that we suggest is based on this simple idea, but instead of tracking a single path, it tracks two. 
Regarding the first obstacle, we now have a natural candidate for the selected user --- the first intersection of the two paths. Regarding the second obstacle, if both paths have ended without intersecting, we simply re-execute the process with a modification that all the already tracked users cannot be selected (this is needed for the incentive compatibility). Note that the (informally) presented mechanism is very simple for implementation in practical settings and it requires only partial knowledge about the network (i.e., full revelation of information by all users is not needed).

Our results are as follows. In Proposition~\ref{pro:ic} we show that the Two Path mechanism is indeed incentive compatible on DAGs. In  Theorem~\ref{thm: tree} we show that the Two Path mechanism performs very well on trees and achieves an approximation ratio of 1.5.
In Theorem~\ref{thm: forests} we show that for forests which are \emph{balanced} (see Definition~\ref{dfn: balance}) the Two Path mechanism achieves a constant approximation ratio\footnote{The balanceness condition is necessary as demonstrated in Example \ref{ex:forest}.}. However, on DAGs (even balanced), the Two Path mechanism might not perform so well, as demonstrated in Example~\ref{ex:dag}. In Section~\ref{sec: analytic 2p} we develop a mechanism that is based on the analysis of the distribution induced by Two Path, and achieve a constant approximation ratio for DAGs that are both balanced, and \emph{monotone} (i.e., DAGs in which a user is always more influential than his followers; see Theorem \ref{thm: DAG}). In Section~\ref{sec: general graphs} we extend the Two Path mechanism to an incentive-compatible mechanism for general networks. Test results of this extension and of the Two Path mechanism on simulated scale-free graphs are presented in Section~\ref{sec: tests}.

	\subsection{Related work}
	Our work is a merger of two branches of literature. Namely, the study of incentive-compatible (IC) selection mechanisms and the study of diffusion models in networks. In this paper we offer, for the first time, IC mechanisms which try to maximize the overall diffusion. 
	
	\subsubsection{Incentive-compatible mechanisms}
Incentive compatibility has been studied in many different contexts. The most relevant to ours are the following.
Alon et al. (\cite{AFPT11}) and Fischer and Klimm (\cite{Fischer:2014:OIS:2600057.2602836}) studied incentive compatible mechanisms in networks. The goal there is to select a vertex with a good approximation to the maximal in-degree. The in-degree can be viewed as a "one-step diffusion". In our work, on the other hand, we focus on a more complex diffusion process.
	Holzman and Moulin (\cite{HM13}) introduced an axiomatic approach to the problem of selecting a prize winning paper based on peer-reviews. In their model, each agent reports one other nominee for a prize and the mechanism selects one winner based on these votes. They presented a set of desirable features a selection mechanism should posses, and asked which subsets of these features we may/ may not have together with incentive-compatibility.
Their paper focuses on self-selection of a winner, which again can be viewed as "one-step diffusion".	
	\\
	More recent works with the same theme can be found in Kurokawa et al. (\cite{Kurokawa:2015:IPR:2832249.2832330}), 
	Aziz et al. (\cite{Aziz:2016:SPS:3015812.3015872}) and Tamura (\cite{TAMURA201641}).
	
	\subsubsection{Diffusion in networks}
	The diffusion of information in networks (and more specifically, social networks) has been extensively studied ever since the seminal paper of Kempe, Kleinberg and Tardos (\cite{KKT}). This literature focuses on maximization of diffusion in a \emph{given} network. In our work we assume that the designer is unaware of the network structure, and asks this information from the users. 
	In \cite{KKT}, the authors have presented two diffusion models --- linear threshold and independent cascade, and offered an algorithm to select a $ k $-subset of the vertices. Their purpose was to show an algorithmically-efficient (i.e., polynomial) mechanism with a finite approximation ratio, but not necessarily IC.\\
	Our influence model relates the influence of vertex $ x $ over vertex $ y $ to the probability of reaching $ x $ in a random walk starting at $ y $. Using random walks to measure popularity/influence is not new either: Gualdi, Medo and Zhang (\cite{0295-5075-96-1-18004}) used the same influence model to rank the influence of academic papers (in their model the graph is acyclic); the popular search engine, Google, uses random walks in its ranking algorithm, called PageRank (\cite{Brin:1998:ALH:297810.297827}); and Andersen et al. (\cite{Andersen:2008:TRS:1367497.1367525}) used random walk in their suggestion of a trust-based recommendation system.\\
	Further reading on networks, their structures and dynamics can be found in the monograph of Easley and Kleinberg (\cite{EK10}) and in \cite{Bornholdt:2003:HGN:640635}.

	\section{The model}\label{sec: diffusion definitions}
	
	\subsection{Network and Diffusion}
	
	A network is a directed graph $ G(V,E) $, where a vertex is a Twitter user (academic paper) and an edge from $ y $ to $ x $ means that $ y $ follows (quotes) $ x $. We denote the in-neighbours (followers) and the out-neighbours (followees) of vertex $ x $ by
	\begin{displaymath}
	N(x):=\{v\in V| (v,x)\in E \};\;
	F(x):=\{v\in V| (x,v)\in E \},
	\end{displaymath}

	respectively, and its out-degree by $ d_o(x)=|F(x)| $.
Our notion of diffusion is defined as follows. We denote by $ \cl{P}_{y,x} $ the family of all simple paths (no vertex repetition) from $ y $ to $ x $. The influence of $ x $ on $ y $ is defined to be
	\begin{displaymath}
	 I(x,y):= \begin{cases}
	\sum\limits_{P\in\cl{P}_{y,x}} \,\prod\limits_{(i,j)\in P}\,\dfrac{1}{d_o(i)}, &x\neq y\\
	1, &x=y
	\end{cases}, 
	\end{displaymath}	
	and the total influence of $ x $ is
	\begin{displaymath}
I(x):=\sum_{y\in V} I(x,y) .
	\end{displaymath}
	As we explain below, this diffusion model neatly relates the influence of a user to the probability of reaching it in a random walk, and is closely related to other well studied models of diffusion.	
	 
The rational behind our notion of diffusion is as follows. 

If $ y $ follows only $ x $, then he is a groupie fan of $ x $ and there is a high probability that he will be affected by any trend introduced by $ x $. If on the other hand, $ y $ follows $ x $ and some other 99 users, then $ x $'s influence over $ y $ is much lower. Moreover, if $ x $ influences $ y $ and $ y $ in turn influences $ z $, then $ x $ has some indirect influence over $ z $. Concretely, we assume that each user divides its attention uniformly between those he follows. A message from user $ x $ diffuses along each path \emph{backwards} with probability equal to the multiplication of the `attention' on the edges.

To gain some intuition about the notion of diffusion, we demonstrate an example of a network and calculate the influence of each vertex. 
	\begin{example}\label{exm: exmample}
		Consider the following Twitter network with six nodes.
			
		\begin{center}
		\begin{tikzpicture}[scale=1, line width=0.4mm]

		\node[label={[font=\normalsize, blue]},circle,draw, text=black](A) at (0,0){A};
		\node[label={[font=\normalsize, blue]},circle,draw, text=black](B) at (40pt,0){B};
		\node[label={[font=\normalsize, blue]},circle,draw, text=black](C) at (20pt,-30pt){C};
		\node[label={[font=\normalsize, blue]},circle,draw, text=black](D) at (60pt,-30pt){D};
		\node[label={[font=\normalsize, blue]},circle,draw, text=black](E) at (90pt,-30pt){E};
		\node[label={[font=\normalsize, blue]},circle,draw, text=black](F) at (45pt,-60pt){F};

		\draw[->] (C) -- node[anchor=north, font=\normalsize,color=blue] {$ \tfrac{1}{3} $}(A);
		\draw[->] (C) -- node[anchor=east, font=\normalsize,color=blue] {$ \tfrac{1}{3} $} (B);
		\draw[bend left=20, ->] (C) to node[ font=\normalsize,color=blue] {$ \tfrac{1}{3} $}(F);
		\draw[->] (D) -- node[anchor=south, font=\normalsize,color=blue] {1}(C);
		\draw[->] (E) -- node[anchor=north, font=\normalsize,color=blue] {$ \tfrac{1}{2} $}(D);
		\draw[->] (E) -- node[anchor=south, font=\normalsize,color=blue] {$ \tfrac{1}{2} $}(B);
		\draw[bend left=20, ->] (F) to node[anchor=east, font=\normalsize,color=blue] {$ \tfrac{1}{2} $} (C);
		\draw[->] (F) -- node[anchor=west, font=\normalsize,color=blue] {$ \tfrac{1}{2} $}(D);

		\end{tikzpicture}
		\end{center}
%		\quad\\
		The edges weights in blue denote the `attention fraction' of this link. For example, node $ E $ follows $ D $ and $ B $, thus each link has a weight of 1/2. Suppose that node $ A $ posts a tweet with a recommendation for a new product. Our model assumes that with probability 1/3 node $ C $ will be affected. If $ C $ is affected, then with probability 1 $ D $ is affected. If $ D $ is affected, then with probability 1/2 $ E $ is affected. Node $ F $ will either be affected directly from $ C $ or from $ C $ through $ D $, thus his probability to be affected from $ A $ is 1/3. We get that starting with $ A $, the expected diffusion is $ 2\tfrac{1}{6} $ (1 for $ A $, $ 1/3 $ for $ C,D,F $ and $ 1/6 $ for $ E $).	Similarly, the influence of $ C $:	
\begin{flalign*}
	&I(C,A)=I(C,B)=0;\; I(C,C)=I(C,D)=1;\;I(C,E)=1\cdot\dfrac{1}{2}=\dfrac{1}{2};\\
	&I(C,F)=\dfrac{1}{2}+1\cdot\dfrac{1}{2}=1;
	\Rightarrow I(C)=3\tfrac{1}{2}.	
\end{flalign*}
The influence of the rest of the nodes:
\[ I(B)=2\tfrac{2}{3};\; I(D)=2;\; I(E)=1;\; I(F)=1\tfrac{5}{6}. \]		
	\end{example}

	We note here the relation of our influence measurement to other popular measures.
	\begin{enumerate}
		\item A \emph{random path} is a path generated by selecting a vertex uniformly at random and `walking' a random walk; each time selecting an out-edge uniformly at random, and stopping when we reach a previously visited vertex or a vertex with no out-edges. We can equivalently define $ I(x) $ to be the probability of visiting $ x $ in a random path, multiplied by $ |V| $.
		\item The \emph{progeny} of vertex $ x $ is the number of vertices which has a path to $ x $ (excluding $ x $). If $ G $ is a forest, then $ I(x) $ is the progeny of $ x $ plus one. Moreover, for any $ G $, let $ G'\subseteq G $ be a random graph generated by picking for each vertex, uniformly at random, one of its out-edges, and removing the rest of its edges. Then $ I(x) $ is the expected progeny of $ x $ in $ G' $ plus one. 
		\item Google's \emph{PageRank} (\cite{Brin:1998:ALH:297810.297827}) is an algorithm which takes as input a digraph $ G $, and a damping factor $ d\in[0,1] $, and outputs a probability distribution on $ V(G) $; this distribution represents the likelihood that an infinite random walk will arrive at any particular vertex. At each step, with probability $ d $ the random walk continues, and with probability $ 1-d $ it jumps to a random vertex. Denote the PageRank value of $ x $ with damping factor 1 by $ PR(x) $. Then $ PR(x) $ can be expressed as (\cite{PRWiki}) 
		\[ PR(x)=\sum_{y\in N(x)}\dfrac{PR(y)}{d_o(y)}. \]		
		\\
		If $ G $ is an acyclic directed graph (DAG), then we can relate the influence of $ x $ to the influence of its neighbours,
		\[ I(x)=1+\sum_{y\in N(x)}\dfrac{I(y)}{d_o(y)}. \]
		The similarity of these two equations is visible. For example, we can deduce that they induce the same ranking. For any $ k $ define $ I^{(k)}(x)=I(x)/k $. Clearly $ I^{(k)} $ induces the same ranking as $ I $; but  
		\[ I^{(k)}(x)=\dfrac{1}{k}+\sum_{y\in N(x)}\dfrac{I^{(k)}(y)}{d_o(y)}, \]
		and as $ k\to\infty $, this definition converges to that of $ PR $.
		\item Finally, consider the \emph{Independent Cascade} model for diffusion (\cite{Goldenberg2001,citeulike:6999265,KKT}). Independent Cascade is defined for weighted digraphs with weights in the interval $ [0,1] $. If we take a random graph $ G'\subseteq G $ by taking each edge independently with probability equal to its weight, then the Independent-Cascade diffusion measure of $ x $ is the expected progeny of $ x $ in $ G' $ plus 1. To see the difference from our model, consider the graph with two vertices $ x,y $, and two edges from $ y $ to $ x $, each with weight $ 1/2 $. Our model will give $ x $ an influence value of $ 2 $ (one for itself and one for $ y $), while the Independent-Cascade score of $ x $ is $ 1\dfrac{3}{4} $ (since there is only 3/4 chance that it will reach $ y $).
	\end{enumerate}

	\subsection{Incentive Compatibility}
	
	\noindent Next, we define what is a selection mechanism and the properties we will be interested in.
	\begin{definition}\label{dfn: mechanism}
		A \emph{selection mechanism} $ \cl{M} $ is a function which gives for every $ G(V,E) $ a probability distribution on $ V\cup \{\emptyset\} $.
	\end{definition}		 \noindent The empty-set value, $ \emptyset $, means that the mechanism has not selected any vertex. We denote by $ \Pr(\cl{M}(G)=x) $ the probability that the mechanism picks $ x\in V\cup \{\emptyset\} $ when the input is $ G $, and by $ \EE[\cl{M}(G)]=\EE[I(\cl{M}(G))] $ the expected influence of the selected vertex\footnote{We take $ I(\emptyset)=0 $ in the calculation of the expectation.}. When it is clear from the context what is the graph $ G $, we sometimes just write $ \Pr(\cl{M}=x)$ and $ \EE[\cl{M}] $. Let $ I^*=\max\limits_{v\in V(G)}I(v) $ be the \emph{maximal influence} in $ G $, and $ v^*=\{v\in V(G): I(v)=I^* \} $ be the set of \emph{optimal vertices}.\\
	Assume that $ \Pr(\cl{M}=x) $ is also the payoff function of $ x $. We would like our mechanism to be such that for any $ x\in V(G) $, it is a best action to report its true out-edges. In addition, we would like our mechanism to give a bounded ratio between the maximal influence and the expected influence of the selected vertex. Our main mechanism will be intended for the subfamily of directed graphs (DAGs); meaning, that only when the reported graph is in this subfamily, we require that it is IC and has bounded ratio.
	\begin{definition}
		A selection mechanism $ \cl{M} $ for the family of graphs $ \cl{G} $ is:
		\begin{itemize}[leftmargin=\dimexpr 26pt]
			\item \emph{incentive-compatible} (IC), if $\forall G\in\cl{G} $ and $ \forall x\in V(G) $, $ \Pr(\cl{M}(G)=x)=\max\limits_{G'\in\cl{G}_x} \Pr( \cl{M}(G')=x) $, where $ \cl{G}_x $ is the family of all graphs we get from $ G $ by adding and removing outgoing edges of $ x $\footnote{Note that we do not require $ \cl{G}_x\subseteq\cl{G} $.};
			\item \emph{efficient} with approx. ratio $ R $, if $ \forall G\in\cl{G}, \dfrac{I^*}{\EE[ \cl{M}(G)]}\leq R $. 
		\end{itemize}		
	\end{definition}

	\noindent  We consider the following four nested families.
	\begin{definition}
		Let $ G $ be a directed graph.
		\begin{itemize}[leftmargin=\dimexpr 26pt]
			\item $ G $ is a \emph{tree} if there is a unique vertex, called the \emph{root}, with no out edges and the rest of the vertices have precisely one out-edge.
			\item $ G $ is a \emph{forest} if it is a disjoint union of trees.
			\item $ G $ is \emph{monotone}\footnote{Observe that a forest is always monotone, and that a monotone graph is acyclic.} if for any edge $ (x,y)\in E(G) $,  $ I(x)<I(y) $.
			\item $ G $ is a \emph{directed acyclic graph (DAG)} if it contains no cycles. 
			
		\end{itemize}
	\end{definition}

	\section{The Two Path Mechanism}\label{sec: algorithmic 2p}
	We will now present the algorithm of the Two Path mechanism, which we denote $ \cl{M}_{2p} $. A \emph{random path} is a random walk which starts at a random vertex, and ends when we reach a vertex with no out-edges or when we return to a previously visited vertex. The idea of the Two Path mechanism is to start two independent random paths from two randomly chosen vertices.	If they intersect on an `unmarked' vertex, it is selected; if they intersect on a `marked' vertex, the mechanism returns `null'; if they do not intersect, all the vertices in these paths are marked, and the mechanism repeats. After presenting the algorithm we will prove it is IC in the family of DAGs; we will then analyse its approximation ratio in the family of trees and in the family of forests.
		
	\begin{algorithmic}[1]
		\State {$ U\gets \emptyset$}
		\While {$ U\neq V $}
		\State {Pick $ x\in V $ uniformly at random}
		\State {$ P_1\gets  $ random path starting at $ x $}
		\State {Pick $ y\in V $ uniformly at random}	
		\State {$ P_2\gets  $ random path starting at $ y $}
		\If {$ P_1\cap P_2=\emptyset $}
		\State {$ U\gets U\cup P_1\cup P_2 $}
		\Else
		\State {$z\gets$ the first\footnote{First according to the path $ P_1 $.} vertex in $ P_1\cap P_2  $}
		\If {$z\in U $}\State{}
		\Return {$ \emptyset $}
		\Else\State{}
		\Return {$ z $}
		\EndIf 
		\EndIf 
		\EndWhile
	\end{algorithmic}

	\begin{proposition}\label{pro:ic}
		Mechanism $ \cl{M}_{2p} $ is IC in the family of DAGs.
	\end{proposition}
	\begin{proof}
		Notice that a vertex can be selected only in the first stage in which it is queried for its out-edges (afterwards it will either be selected or marked). It is enough, then, to show that in the first time a vertex is reached by one of the random paths, reporting its true edges is a best action. Suppose $ P_1 $ reaches vertex $ v $ and we query for its out-edges for the first time. Vertex $ v $ can only be selected if $P_2$ reaches it before it reaches any other vertex of $ P_1 $. Since $ G $ has no cycles, a true report by $ v $ will lead to vertices which can be reached by $ P_2 $ only after $ P_2 $ has `missed' $ v $. Hence reporting its true edges cannot hurt its chances of being selected. Surely, reporting additional edges cannot help it. If $ P_2 $ reaches an unmarked vertex $ u $, then either $ u $ is already selected (if $ u\in P_1 $) or it will never be selected regardless of the edges it reports. Hence this mechanism is incentive compatible.
	\end{proof}
	Mechanism $ \cl{M}_{2p} $ is generally not IC when the graph contains cycles, as demonstrated by the next example.
	\begin{example}	\label{ex: general graphs}
	Take the graph with two vertices, $ x,y $, and two edges $ (x,y), (y,x) $. Since $ P_1 $ always includes both vertices, the winner will be determined by the starting vertex of $ P_2 $. However, if we remove the edge $ (x,y) $, then $ x $ will be selected when \emph{either} $ P_1 $ or $ P_2 $ starts from $ x $. Thus $ x $ has an incentive to remove its edge to $ y $. 
	\end{example}
	\subsection{Two Path on the family of trees}
		When $ G $ is a tree, every two paths intersect, and mechanism $ \cl{M}_{2p} $ is guaranteed to return a vertex in the first stage. Suppose $ G $ is a path of length $ n $. The mechanism then returns either $ x $ or $ y $, whichever is further along the path. Thus, the expected result in this case is $ 2n/3 $. In the next theorem we show that this is the worse scenario when $ G $ is a tree. Hence, we find that the exact approximation ratio on trees is 1.5. This theorem is fundamental for the proof of the bound of the approximation ratio on forests (Theorem~\ref{thm: forests}).	
	\begin{theorem}\label{thm: tree}
		In the family of all trees, $ R_{\cl{M}_{2p}}=1.5 $.
	\end{theorem}
	\begin{proof}
		For a vertex $ v $ we denote by $ T_v $ the subtree which is under $ v $. Since all the vertices, except the root, have one out-edge, $ I(v)=|T_v| $. Vertex $ v\in V $ is selected if and only if $ x,y\in T_v $ and they are not in the same proper subtree of $ T_v $. We get:
		\begin{flalign*}
		\Pr(\cl{M}_{2p}=v)&=\Pr(x,y\in T_v)-\sum_{u\in N(v)}\Pr(x,y\in T_u)\\
		&=\left(\dfrac{I(v)}{n}\right)^2-\sum_{u\in N(v)}\left(\dfrac{I(u)}{n}\right)^2.\\
		\EE[\cl{M}_{2p}]&=\sum_{v\in V}I(v)\Pr(v)=\dfrac{1}{n^2}\sum_{v\in V}\left (I^3(v)-I(v)\sum_{u\in N(v)}I^2(u) \right)\\
		&=\dfrac{1}{n^2}\left(\sum_{v\in V}I^3(v)-\sum_{v\in V\backslash\text{root}}I^2(v)\cdot I(F(v)) \right)\\
		&=\dfrac{1}{n^2}\left (n^3-\sum_{v\in V\backslash\text{root}}I^2(v)(I(F(v))-I(v)) \right ).
		\end{flalign*}
		We define the function $ f $ from the family of all trees of order $ n $ to $ \NN_+ $ by: 
		\[ f(T)= \sum_{v\in V(T)\backslash\text{root}}I^2(v)(I(F(v))-I(v)). \]
		\begin{lemma}\label{lem: path}
			Let $ \cl{P}_n $ be the path with $ n $ vertices. Then,  
			\[\cl{P}_n=\argmax\limits_{T \text{ is a tree of order } n}f(T).\]
		\end{lemma}
		We will first show that the lemma completes the proof. Indeed, if $ f(G) $ is maximized when $ G $ is a path, then $ \EE[\cl{M}_{2p}] $ is minimized in this case. When $ G $ is a path, the mechanism will always return either $ x $ or $ y $, whichever is further along the path. The expected value of the mechanism in this case is, therefore, the expected value of $ \max(a,b) $, where $ a,b $ are two independent, integral random variables, uniformly distributed in $ [1,n] $. We get
		\begin{flalign}
		\EE[\cl{M}_{2p}]=\sum_{i=1}^{n}i\left(\dfrac{2i-1}{n^2} \right) =\dfrac{1}{n^2}\cdot\dfrac{n(n+1)(4n-1)}{6}\geq \dfrac{2}{3}n=\dfrac{2}{3}I^*.\label{ineq: tree expectation}
		\end{flalign}
		In fact, $ \EE[\cl{M}_{2p}]/n\xrightarrow[]{n\to\infty} \dfrac{2}{3} $, hence we have found the exact value of $ R_{M_{2p}} $ in trees.
		\end{proof}
		 \noindent It remains to prove the lemma.
		\begin{proof}[Proof of Lemma~\ref{lem: path}]\belowdisplayskip=-12pt
			Let $ G $ be any tree of order $ n $ which maximizes $ f(T) $. Assume for contradiction that $ G $ is not a path. Then there is a vertex $ v $ such that there are two paths: $ A=\{a_1,\ldots, a_k,v\} $, $ B=\{b_1,\ldots, b_\ell, v\} $ such that $ a_1,b_1 $ are leafs and $ \{a_2,\ldots, a_k, b_2,\ldots,b_\ell \} $ all have in-degrees one. Let $ G' $ be the tree in which we remove the edge $ (a_k, v) $ and add the edge $ (a_k,b_1) $. It is enough to show that $ f(G')> f(G) $. Notice that the only vertices which have different contribution to the sums in $ f(G) $ and $ f(G') $ are $ a_k $ (since $ I(F(a_k)) $ is different) and $ b_1,\ldots,b_l $. Hence,
			\begin{flalign*}
			&f(G)-f(G')=[I^2(a_k)I(F(a_k))]|_{G}-[I^2(a_k)I(F(a_k))]|_{G'}\\
			&\qquad+\sum_{i=1}^{\ell}\left \{[I^2(b_i)(I(F(b_i))-I(b_i))]|_G-[I^2(b_i)(I(F(b_i))-I(b_i))]|_{G'}\right \}\\
			&=[I^2(a_k)I(v)]|_{G}-[I^2(a_k)I(b_1)]|_{G'}\\
			&\qquad+\sum_{i=1}^{\ell-1}\left \{[I^2(b_i)(I(b_{i+1})-I(b_i))]|_G-[I^2(b_{i})(I(b_{i+1})-I(b_i))]|_{G'}\right \}\\
			&\qquad+[I^2(\ell)(I(v)-I(\ell))]|_G-[I^2(\ell)(I(v)-I(\ell))]|_{G'}\\
			&=k^2(k+\ell+1)-k^2(k+1)\\
			&\qquad+\sum_{i=1}^{\ell-1}\left [\left (i^2(i+1-i)-(k+i)^2(k+i+1-(k+i)\right)\right ] \\
			&\qquad+\ell^2(k+\ell+1-\ell)-(k+\ell)^2(k+\ell+1-k-\ell)\\
			&=k^2\ell-\sum_{i=1}^{\ell-1}(k^2+2ki)+\ell^2(k+1)-(k+\ell)^2\\
			&=k^2-k\ell(\ell-1)+\ell^2(k+1)-(k+\ell)^2\\
			&=-k\ell<0. 
			\end{flalign*}
		\end{proof}

	\subsection{Two Path on the family of forests}
	Let $ G $ be a forest. Denote by $ S\subseteq V(G) $ the set of roots of $ G $. We denote $ s=|S| $. Suppose $ s $ is `large', e.g. $ s=\sqrt{n} $. Then in a single stage, there is a high probability that the two random paths will be in different trees and those two paths will be marked. We claim, however, that if the distribution of the orders of the trees is reasonably concentrated near the average, there is a positive probability that the first time the two paths intersect, they will be intersected in a tree in which no vertex is marked. This will imply, together with Theorem~\ref{thm: tree}, a bound on the approximation ratio. To be precise, define
	\[ \overline{I}=\dfrac{\sum_{r\in S}I(r)}{s}=\dfrac{n}{s}, \]
	the average influence of the sinks. We will prove the following.
	\begin{theorem}\label{thm: forests}
		For any forest $ G $,
		\begin{flalign*}
			\EE[\cl{M}_{2p}(G)]\geq 0.09\dfrac{n}{s}.
		\end{flalign*}
	\end{theorem}
	We define the following measure for the balance of a forest.
\begin{definition}\label{dfn: balance}
	For $ \alpha\in (0,1]$, a forest $ G $ is \emph{$ \alpha $-balanced} if,
	\[ \dfrac{\overline{I}}{I^*}=\dfrac{n/s}{I^*}\geq\alpha. \]
\end{definition}
This definition formally captures the idea of `reasonable distribution' of the trees' orders mentioned above. Thus Theorem~\ref{thm: forests} implies the following bound on the approximation ratio.
	\begin{corollary}
		In the family of $ \alpha $-balanced forests,
		\[ R_{M_{2p}}\leq 1/0.09\alpha. \]
	\end{corollary}
 Note that our purpose here is to show that the approximation ratio can be finitely-bounded using a natural parameter of the graph. Although our \emph{theoretical} bound of $ 1/\alpha $ might be quite high, our simulations in Section~\ref{sec: tests} demonstrate that the actual approximation ratio in natural classes of networks is usually low.\\
 Before turning to the proof, we show in the next example that indeed when $ \alpha\to 0 $ the approximation ratio of $ \cl{M}_{2p} $ cannot be bounded.
\begin{example}\label{ex:forest}
	
	Consider a forest made of one star of order $ \sqrt{n} $ and $ n $ isolated vertices. \\
	
	\begin{tikzpicture}[scale=1, line width=0.4mm, every node/.style={draw,circle,inner sep=0pt, minimum size=20pt}]

	\node[](x) at (30pt,0pt){$ x $};
	\node[](v1) at (0,-30pt){$ 1 $};
	\node[](v2) at (25pt,-30pt){$ 2 $};
	\node[](vm) at (60pt,-30pt){$ \sqrt{n} $};
	
	\node[](u1) at (100pt,-30pt){$ 1 $};
	\node[](u2) at (125pt,-30pt){$ 2 $};
	\node[](un) at (160pt,-30pt){$ n $};
	
	\draw[->] (v1) -- (x);
	\draw[->] (v2) -- (x);
	\draw[->] (vm) -- (x);
	\draw[dotted] (38pt,-30pt) -- (47pt,-30pt);
	\draw[dotted] (138pt,-30pt) -- (147pt,-30pt);
	
	\end{tikzpicture}
	
	\noindent The centre of the star, $ x $, will only be selected if both paths hit the star for the first time on the same stage. With high probability, this event will not happen\footnote{At the first round the probability that $ x $ will be selected is $ O(n^{-1}) $, whereas the probability that $ x $ will be marked is $ O(n^{-1/2}) $. Same is true for all the first $ n^{3/4}  $ rounds. Therefore, the probability that $ x $ will be selected during at the first $ n^{3/4} $ rounds is $ o(n^{-1/4}) $, while the probability that it will be marked is close to 1.}; hence, the result will be either $ \emptyset $ or a vertex with influence one. 
\end{example}

	\begin{proof}[Proof of Theorem~\ref{thm: forests}]
		For any $ r\in S $, we denote by $ T_r $ the tree of $ r $. We define $ A_r $ to be the event that the mechanism  has picked a vertex from $ T_r $ without marking $ r $ before. That is, $ A_r $ is the event that the first time the two paths meet, they meet in $ T_r $, and none of the paths in previous stages was in $ T_r $. From Theorem~\ref{thm: tree} we know that
		\[ \EE[\cl{M}_{2p}|A_r]\geq \dfrac{2}{3}I(r). \]
		Denote the probability that the two paths intersect in a single stage by $ Z=\Pr(P_1\cap P_2\neq\emptyset)=\sum_{r\in S}(I(r)/n)^2 $. Let $ q_r $ be the probability that in a single stage, the two paths did not intersect and none of them was in $ T_r $. Then,
		\[ q_r\geq (1-Z)(1-2I(r)/n)\geq 1-Z-2I(r)/n, \]
		and,
		\[ \Pr(A_r)=\sum_{k=0}^\infty q_r^k\left (\dfrac{I(r)}{n}\right )^2=\dfrac{I^2(r)}{n^2(1-q_r)}\geq\dfrac{I^2(r)}{n^2Z+2nI(r)}. \]
		Now, since the events $ \{A_r\}_{r\in S} $ are pairwise disjoint, we may bound
		\begin{flalign}
		\EE[\cl{M}_{2p}]\geq\sum_{r\in S}\Pr(A_r)\EE[\cl{M}_{2p}|A_r]\geq \dfrac{2}{3}\sum_{r\in S}\dfrac{I^3(r)}{n^2Z+2nI(r)}.\label{eq: sum forests}
		\end{flalign}
		In order to bound the last sum, we partition $ S $ to two parts:
		\[ S_1=\{r\in S: Z\leq I(r)/n \};\; S_2=S\backslash S_1, \]
		and bound this sum for each part separately.
		\begin{flalign*}
		&\dfrac{2}{3}\sum_{r\in S_1}\dfrac{I^3(r)}{n^2Z+2nI(r)}\geq\dfrac{2}{3}\sum_{r\in S_1}\dfrac{I^3(r)}{n^2I(r)/n+2nI(r)}=\dfrac{2n}{9}\sum_{r\in S_1}\left(\dfrac{I(r)}{n}\right)^2;\\
		&\dfrac{2}{3}\sum_{r\in S_2}\dfrac{I^3(r)}{n^2Z+2nI(r)}\geq\dfrac{2}{3}\sum_{r\in S_2}\dfrac{I^3(r)}{n^2Z+2n^2Z}\geq\dfrac{2n}{9Z}\sum_{r\in S_2}\left (\dfrac{I(r)}{n}\right )^3\\
		&\geq \dfrac{2n}{9Z\sqrt{s}}\left(\sum_{r\in S_2}\left (\dfrac{I(r)}{n}\right )^2\right)^{3/2}.
		\end{flalign*}
		Where the last inequality is due to the convexity of $ x^{3/2} $. Let $ \delta $ be the solution of $ (1-\delta)^{3/2}=\delta $. Since $ Z=\sum_{r\in S}(I(r)/n)^2 $, either $ \sum_{r\in S_1}(I(r)/n)^2\geq\delta Z $ or $ \sum_{r\in S_2}(I(r)/n)^2\geq(1-\delta) Z $. In the former case we bound
		\[ 	\EE[\cl{M}_{2p}]\geq\dfrac{2n}{9}\sum_{r\in S_1}\left(\dfrac{I(r)}{n}\right)^2\geq \dfrac{2n\delta Z}{9}, \]
		and in the latter case we bound
		\[ \EE[\cl{M}_{2p}]\geq\dfrac{2n}{9Z\sqrt{s}}\left(\sum_{r\in S_2}\left (\dfrac{I(r)}{n}\right )^2\right)^{3/2}\geq \dfrac{2n((1-\delta)Z)^{3/2}}{9Z\sqrt{s}}=\dfrac{2n\delta \sqrt{Z}}{9\sqrt{s}}. \]
		Since $ \sum_{r\in S}I(r)=n $, we can use convexity  to bound $ Z $,
		\[ Z=\sum_{r\in S}\left(\dfrac{I(r)}{n}\right)^2\geq \dfrac{1}{s}\left (\sum_{r\in S}\dfrac{I(r)}{n} \right )^2=\dfrac{1}{s}.\]
		Hence, we get that in any case,

		\[
		\EE[\cl{M}_{2p}]\geq\dfrac{2\delta n }{9s}\approx \dfrac{2\cdot 0.43\cdot n}{9s} \geq 0.09\dfrac{n}{s}.\qedhere\]
		
	\end{proof}

\section{Mechanism for monotone DAGs}\label{sec: analytic 2p}
We remind that we define a monotone graph to be a graph in which any user is more influential than his followers. Clearly, a monotone graph must be acyclic. Monotonicity is a natural property in domains where the statement "my friend is more influential than I am" is true for any vertex. \\
Let $ G $ be a DAG. Denote by $ S\subseteq V(G) $ the set of \emph{sinks} of $ G $, i.e. the set of vertices with no out-edges. Denote $ s=|S| $. In a forest, $ S $ is the set of roots and, as we proved in Theorem~\ref{thm: forests}, $ \EE[\cl{M}_{2p}(G)]\geq cn/s $ for a constant $ c>0 $. The next example shows that this claim is not true for all monotone DAGs. We will show later in this section a mechanism which is somewhat related to $ \cl{M}_{2p} $, and which generalizes Theorem~\ref{thm: forests} for monotone DAGs.
\begin{example}\label{ex:dag}
	Consider a matrix of $ n^{3/4}\times n^{1/4} $ vertices and another vertex, $ v_0 $. Suppose each vertex in row $ 1\leq i\leq n^{3/4}-1 $ has an edge to every vertex in row $ i+1 $, and the vertices of row $ n^{3/4} $ all have a single edge to vertex $ v_0 $. 

	\begin{center}
	\begin{tikzpicture}[scale=1, line width=0.1mm, every node/.style={draw,circle,inner sep=0pt, minimum size=30pt}]

	\node[](v0) at (45pt,0pt){$ v_0 $};
	\node[](v1) at (0,-35pt){$ n^{3/4},1 $};
	\node[](v2) at (40pt,-35pt){$ n^{3/4},2 $};
	\node[](vm) at (90pt,-35pt){\tiny $ n^{3/4},n^{1/4} $};
	
	\node[](v12) at (0,-80pt){$ 2,1 $};
	\node[](v22) at (40pt,-80pt){$ 2,2 $};
	\node[](vm2) at (90pt,-80pt){$ 2,n^{1/4} $};
	
	\node[](v11) at (0,-115pt){$ 1,1 $};
	\node[](v21) at (40pt,-115pt){$ 1,2 $};
	\node[](vm1) at (90pt,-115pt){$ 1,n^{1/4} $};

	\begin{scope}[on background layer]
	\draw[->] (v1) -- (v0);
	\draw[->] (v2) -- (v0);
	\draw[->] (vm) -- (v0);
	
	\draw[->] (v11) -- (v12);
	\draw[->] (v11) -- (v22);
	\draw[->] (v11) -- (vm2);
	\draw[->] (v21) -- (v12);
	\draw[->] (v21) -- (v22);
	\draw[->] (v21) -- (vm2);
	\draw[->] (vm1) -- (v12);
	\draw[->] (vm1) -- (v22);
	\draw[->] (vm1) -- (vm2);
	
	\draw[->] (v12) -- (5pt,-60pt);
	\draw[->] (v12) -- (40pt,-60pt);
	\draw[->] (v12) -- (90pt,-60pt);
	\draw[->] (v22) -- (5pt,-60pt);
	\draw[->] (v22) -- (40pt,-60pt);
	\draw[->] (v22) -- (90pt,-60pt);
	\draw[->] (vm2) -- (5pt,-60pt);
	\draw[->] (vm2) -- (40pt,-60pt);
	\draw[->] (vm2) -- (90pt,-60pt);

	\draw[dotted] (v2) -- (vm);
	\draw[dotted] (v22) -- (vm2);
	\draw[dotted] (v21) -- (vm1);
	\draw[dotted] (v12) -- (v1);
	\end{scope}
	\end{tikzpicture}
	\end{center}
	
	\noindent The vertices of row $ i $ all have influence $ i $, and vertex $ v_0 $ has influence $ n+1 $. Thus, this graph is monotone. Since $ v_0 $ is the only sink, $ \overline{I}=(n+1)/s=n+1=I^* $. However, $ \Pr(\cl{M}_{2p}(v_0))\xrightarrow[]{n\to\infty} 0 $. Indeed, with high probability, both $ x,y $ will be somewhere in the first $ n^{3/4}-n^{1/2} $ rows of the matrix. Since, for each of the top $ n^{1/2} $ rows the random paths $ P_1, P_2 $ have an independent probability of $ n^{-1/4} $ to intersect, we get that with high probability the two paths will intersect before reaching $ v_0 $. We therefore get that for this monotone DAG, the Two Path mechanism does not have a bounded approximation ratio.
\end{example}

The mechanism which we are about to suggest for monotone DAGs, will not be described as an algorithmic procedure, but rather as an explicit distribution formula. We obtain this formula by first finding an explicit expression for the distribution of $ \cl{M}_{2p} $, and then generalizing it in a natural way to monotone DAGs. We start by finding the distribution of $ \cl{M}_{2p} $ when $ G $ is a tree. In this case, the probability of two independent random paths to intersect in $ v\in V(G) $ is precisely $ (I(v)/n)^2 $. However, this is not the probability $ \Pr(\cl{M}_{2p}=v) $, because $ v $ is only selected if it is the \emph{first} intersection of the two paths. Define the recursive function:
\begin{flalign}
Z(v)=Z(v,G):=\left(\dfrac{I(v)}{n}\right)^2-\sum_{u\in P(v)}Z(u), \label{eq: Z} 
\end{flalign}
where $ P(v)\subseteq V $ is the progeny set of $ v $, i.e. all vertices which have a path to $ v $ (not including $ v $ itself). We can then write
\[ \Pr(\cl{M}_{2p}=v)=Z(v), \]
and we have found an explicit expression for the distribution induced by $ \cl{M}_{2p} $. We remark that it is not hard to prove, using simple induction, that
\begin{flalign}
Z(v)=\left(\dfrac{I(v)}{n}\right)^2-\sum_{u\in N(v)}\left(\dfrac{I(u)}{n}\right)^2.  \label{eq: Z(v)}
\end{flalign}

\noindent Now suppose that $ G $ is a forest. Observe the following:
\begin{itemize}
	\item The probability that $ v\in V $ is selected in the first stage is $ Z(v) $.
	\item In every subsequent stage, $ v $ will have probability $ Z(v) $ to be selected, provided the two paths did not intersect in previous stages and $ v $ was not marked.
\end{itemize}
Let $ Z=Z(G):=\sum_{u\in V}Z(u)=\sum_{r\in S(G)}(I(r)/n)^2$ denote the probability that in a single stage the two paths intersect. Let $ G_v $ be the graph we get from $ G $ after removing all the out-edges of $ v $. Then $ Z_v=Z(G_v) $ is the probability that in a single stage the Two Path intersected but none of them went through $ v $, unless $ v $ is the intersection vertex (this is because $ v $ is a root vertex in $ G_v $). Thus, if we denote by $ q(v) $
the probability that in a single stage the two paths did not intersect and none of them went through $ v $, then
\begin{flalign*}
q(v)\geq(1-Z_v)(1-2I(v)/n).
\end{flalign*}
We conclude that 
\begin{flalign}
\Pr(\cl{M}_{2p}=v)=\sum_{k=1}^\infty q(v)^kZ(v)=\dfrac{Z(v)}{1-q(v)}\geq \dfrac{Z(v)}{Z_v+2\tfrac{I(v)}{n}}. \label{eq: forests dist}
\end{flalign}
We will now use this last expression as a baseline for our `analytic' Two Path mechanism, denoted $ \cl{M}_{2p}^A $.

Let $ G $ be a DAG. In a forest the influence of one vertex over the other, $ I(v,u) $, is either 1 or 0. This is no longer the case for DAGs. To account for this difference we alter the function we defined at~(\ref{eq: Z}):
\begin{flalign*}
&Z(v)=Z(v,G)=\left(\dfrac{I(v)}{n}\right)^2-\sum_{u\in N(v)}I(v,u)\left(\dfrac{I(u)}{n}\right)^2;\\
&Z=Z(G)=\sum_{u\in V}Z(u)=\sum_{r\in S(G)}\left(\dfrac{I(r)}{n}\right)^2.
\end{flalign*}
Now, our mechanism for monotone DAGs, denoted $ \cl{M}_{2p}^A $, is defined by the following distribution:
\begin{flalign}
\Pr(\cl{M}_{2p}^A(G)=v)=\dfrac{Z(v)}{Z_v+2\tfrac{I(v)}{n}}, \label{eq: monotone mechanism}
\end{flalign}
when $ G $ is a monotone DAG; if $ G $ is not a monotone DAG, the mechanism returns $ \emptyset $.

\begin{proposition}\label{prp: monotone}
Mechanism $ \cl{M}_{2p}^A$ is well-defined and incentive compatible in the family of monotone DAGs.
\end{proposition}
We will need the following lemma.
\begin{lemma}\label{lem: DAG}
	For any monotone DAG $ G $ and for every $ v\in V(G) $, 
	\[ Z-2\tfrac{I(v)}{n}\leq Z_v\leq Z .\]
\end{lemma}
\begin{proof}
	If $ v $ is a sink then $ Z_v=Z $. Assume then that $ v $ is not a sink. On the one hand, the influence of some of the sinks of $ G $ is lower in $ G_v $, but on the other hand $ v $ is an extra sink which was not in $ G $. Thus,
	\begin{flalign*}
	Z_v&=\sum_{r\in S(G_v)}\left(\dfrac{I(r)|_{G_v}}{n}\right)^2\\
	&=Z-\sum_{r\in S(G)}\left[\left(\dfrac{I(r)}{n}\right)^2-\left(\dfrac{I(r)-I(r,v)I(v)}{n}\right)^2  \right] +\left(\dfrac{I(v)}{n}\right)^2\\
	&=Z-\sum_{r\in S(G)}\dfrac{2I(r)I(r,v)I(v)-I^2(r,v)I^2(v)}{n^2} +\left(\dfrac{I(v)}{n}\right)^2\\
	&=Z-2\dfrac{I(v)}{n}\left(\sum_{r\in S(G)}I(r,v)\dfrac{I(r)}{n}-\dfrac{1}{2}\cdot \dfrac{I(v)}{n}\left (1+\sum_{r\in S(G)}I^2(r,v) \right ) \right).	
	\end{flalign*}
	Now, $\sum_{r\in S(G)}I(r,v)=1\Longrightarrow \sum_{r\in S(G)}I^2(r,v)\leq 1 $, and by monotonicity $ \sum_{r\in S(G)}I(r,v)I(r)\geq I(v)\sum_{r\in S(G)}I(r,v)= I(v) $. Hence, we have,
	\begin{flalign*}
	&\sum_{r\in S(G)}I(r,v)\dfrac{I(r)}{n}-\dfrac{1}{2}\cdot \dfrac{I(v)}{n}\left (1+\sum_{r\in S(G)}I^2(r,v) \right )
	\geq 0,
	\end{flalign*}
	and the upper bound follows. \\
	For the lower bound it is enough to observe that by Cauchy-Schwarz, 
	\[ \sum_{r\in S(G)}I(r,v)\dfrac{I(r)}{n}\leq \sqrt{\sum_{r\in S(G)}I^2(r,v)}\sqrt{\sum_{r\in S(G)}I^2(r)/n^2}\leq 1.\qedhere \]
\end{proof}

\begin{proof}[Proof of Proposition~\ref{prp: monotone}]
	Since it is clear from (\ref{eq: monotone mechanism}) that $ \Pr(\cl{M}_{2p}^A=v) $ does not depend on the out-edges of $ v $, the mechanism is IC.\\
	To prove that it is well-defined, we need to show that for every monotone DAG, all the probabilities are non-negative and the sum of probabilities is at most 1 \footnote{If the sum is strictly less than 1, then the rest of the probability goes to $ \emptyset $.}. For the former, we use the monotonicity assumption, which means that $ I(v)>I(u) $ for all $ u\in N(v) $. Hence,
	\begin{flalign*}
	&Z(v)=\left(\dfrac{I(v)}{n}\right)^2-\sum_{u\in N(v)}I(v,u)\left(\dfrac{I(u)}{n}\right)^2\\
	&>\left(\dfrac{I(v)}{n}\right)^2-\dfrac{I(v)}{n}\sum_{u\in N(v)}I(v,u)\dfrac{I(u)}{n}=\left(\dfrac{I(v)}{n}\right)^2-\dfrac{I(v)(I(v)-1)}{n}>0.
	\end{flalign*}
	For the latter, we use the lemma.
	\begin{flalign*}
	\sum_{v\in V}\Pr(\cl{M}_{2p}^A(G)=v)=\sum_{v\in V}\dfrac{Z(v)}{Z_v+2\tfrac{I(v)}{n}}\leq \dfrac{\sum_{v\in V}Z(v)}{Z}=1.\qedhere
	\end{flalign*}	
	
\end{proof}
We are ready to prove the parallel of Theorem~\ref{thm: forests}.
\begin{theorem}\label{thm: DAG}
	Let $ G $ be a monotone DAG with $ s $ sinks, $ S\subseteq V(G) $. Then,
	\begin{flalign*}
	\EE[\cl{M}_{2p}^A(G)]\geq 0.07\dfrac{n}{s}.
	\end{flalign*}
	
\end{theorem}
\begin{proof}
	Define, for every $ r\in V $,
	\[ \EE_r=\sum_{v\in V}I(r,v)I(v)\Pr(\cl{M}_{2p}^A=v). \]
	Since $ \forall v, \sum_{r\in S}I(r,v)=1 $,
	\begin{flalign*}
	\EE[\cl{M}_{2p}^A]&=\sum_{v\in V}I(v)\Pr(\cl{M}_{2p}^A=v)=\sum_{v\in V}I(v)\Pr(\cl{M}_{2p}^A=v)\sum_{r\in S}I(r,v)\\
	&= \sum_{r\in S}I(r,v)\sum_{v\in V}I(v)\Pr(\cl{M}_{2p}^A=v)= \sum_{r\in S}\EE_r. 
	\end{flalign*}
	It is enough then, to bound $ \EE_r $. Let $ P(r)\subseteq V $ be the progeny of $ r $. We will prove by induction on $ |P(r)| $ that 
	\[ \EE_r\geq \dfrac{1}{2}\cdot\dfrac{I^3(r)}{n^2Z+2nI(r)}. \]
	When $ P(r)=\emptyset $, $ I(r)=1 $ and by Lemma~\ref{lem: DAG},
	\begin{flalign*}
	\EE_r=\Pr(\cl{M}_{2p}^A=v)= \dfrac{1/n^2}{Z_v+2/n}\geq \dfrac{1}{n^2Z+2n}.
	\end{flalign*}
	In the general case, we write
	$ \EE_r=I(r)\Pr(\cl{M}_{2p}^A=r)+\sum_{v\in N(r)}I(r,v)\EE_v $.\\ Now,
	\begin{flalign*}
	I(r)\Pr(\cl{M}_{2p}^A=r)=\dfrac{I(r)Z(r)}{Z_r+2\tfrac{I(r)}{n}}\geq\dfrac{I^3(r)-I(r)\sum_{v\in N(r)}I(r,v)I^2(v)}{n^2Z+2nI(r)},
	\end{flalign*}
	while, by induction,
	\begin{flalign*}
	\sum_{v\in N(r)}I(r,v)\EE_v&\geq \dfrac{1}{2}\sum_{v\in N(r)}\dfrac{I(r,v)I^3(v)}{n^2Z_v+2nI(v)}\geq \dfrac{1}{2}\sum_{v\in N(r)}\dfrac{I(r,v)I^3(v)}{n^2Z+2nI(r)}.
	\end{flalign*}
	We get,
	\begin{flalign}
	\EE_r\geq \dfrac{I^3(r)}{n^2Z+2nI(r)}+\sum_{v\in N(r)}\dfrac{I(r,v)I(v)}{n^2Z+2nI(r)}\left[\dfrac{1}{2}I^2(v)-I(r)I(v)\right].
	\end{flalign}
	The term in the parenthesis is minimized when $ I(v)=I(r) $, and
	\begin{flalign*}
	\EE_r&\geq \dfrac{I^3(r)}{n^2Z+2nI(r)}-\dfrac{1}{2}\sum_{v\in N(r)}\dfrac{I(r,v)I(v)I^2(r)}{n^2Z+2nI(r)}\\
	&=\dfrac{I^3(r)}{n^2Z+2nI(r)}-\dfrac{1}{2}\cdot\dfrac{I^2(r)(I(r)-1)}{n^2Z+2nI(r)}\geq\dfrac{1}{2}\cdot\dfrac{I^3(r)}{n^2Z+2nI(r)}.
	\end{flalign*}
	We proved:
	\begin{flalign*}
	\EE[\cl{M}_{2p}^A]\geq \dfrac{1}{2}\sum_{r\in S}\dfrac{I^3(r)}{n^2Z+2nI(r)}.
	\end{flalign*}
	The last sum is precisely the sum we got in the proof of Theorem~\ref{thm: forests} in~(\ref{eq: sum forests}), multiplied by 3/4. We continue in exactly the same manner and get:	

	\[ \EE[\cl{M}_{2p}^A]\geq \dfrac{3}{4}\cdot\dfrac{2}{9}\cdot0.43\cdot \dfrac{n}{s}\geq 0.07\dfrac{n}{s}.\qedhere\]
\end{proof}

\section{Mechanism for general graphs}\label{sec: general graphs}
	Mechanism $ \cl{M}_{2p} $ is defined for all networks, but is incentive-compatible only on the family of DAGs (see Example~\ref{ex: general graphs}). We will now present mechanism $ \cl{M}_{2p}^G $, which is based on $ \cl{M}_{2p} $ and which is IC on all graphs\footnote{Notice, that in this case, reporting the true edges becomes a weakly dominant strategy for all vertices.}. The idea is to choose a random ordering of $ V(G) $, and remove all edges from vertices with high index to vertices with lower index. The resulted graph is clearly acyclic, and we return the outcome of the running of $ \cl{M}_{2p} $ on this graph.
	
	\begin{algorithmic}[1]
		\State {$\succ\gets $ random ordering of $ V(G) $}
		\ForAll {$ e=(x,y)\in E(G) $}
		\If {$ x\succ y $}
		\State {$ E\gets E-e $}
		\EndIf
		\EndFor\State{}
		\Return {$ \cl{M}_{2p}(G) $}
		\end{algorithmic}
		
		\begin{proposition}
			Mechanism $ \cl{M}_{2p}^G $ is IC for all graphs.
			\end{proposition}
			\begin{proof}
				For every graph $ G $ and for every ordering $ \succ $ on $ V(G) $, let $ F(G,\succ) $ be the resulted DAG after removing from $ G $ edges as described above with $ \succ $ as the ordering. Since every ordering has the same probability of $ 1/n! $, we get that
				\[ \Pr(\cl{M}_{2p}^G=v)=\dfrac{1}{n!}\sum_{\substack{\succ \text{ is an}\\ \text{order on }V}}\Pr(\cl{M}_{2p}(F(G,\succ))).  \]
				It is enough, then, to show that the out-edges of $ v $ do not influence $ \Pr(\cl{M}_{2p}(F(G,\succ))) $. Indeed, any edge from $ v $ to a vertex $ u $ such that $ v\succ u $ is not in $ F(G,\succ) $. The out-edges of $ v $ which remain in $ F(G,\succ) $ do not influence $ \Pr(\cl{M}_{2p}(F(G,\succ))) $, since $ \cl{M}_{2p} $ is IC on DAGs. 
				\end{proof}
				Note that $ \cl{M}_{2p}^G $ is not an extension of $ \cl{M}_{2p} $, that is, we cannot claim that in the family of DAGs, $ \cl{M}_{2p}^G=\cl{M}_{2p} $. In fact, it does not even have a bounded approximation ratio in the family of trees. Take, for example, the complete binary tree. After the random ordering and edges-removal, we will get a forest with about $ n/2 $ trees. Thus, the outcome of $ \cl{M}_{2p}^G $ on this tree is close to that of the random mechanism. Since the average influence in the binary tree is $ \log n $, the approximation ratio of $ \cl{M}_{2p}^G $ is at least $ n/\log n $. Nevertheless, in the next section we will test this mechanism on simulated `Twitter-like' networks and see that it gives a reasonable approximation.

\section{Experimental study}\label{sec: tests}
It has been long known that many real-world networks, and social networks in particular, feature the property that their degree distribution follows a power law\footnote{Meaning that the fraction of vertices which have in-degree $ k $ is proportional to $ k^{-\lambda} $ where $ \lambda>1 $ is a parameter of the network.} (see, for example, \cite{Caldarelli07,2002AdPhy..51.1079D}). Networks with this property are known as `scale-free' graphs. Several models have been offered to simulate scale-free graphs in a way which will resemble the constant growth and self-organising nature of social networks. In this section we present test results of mechanism $ \cl{M}_{2p} $ on scale-free DAGs, simulated according to the model of Barab\'asi and Albert (\cite{Barabasi509}); and results of mechanism $ \cl{M}_{2p}^G $ on scale-free general graphs, simulated according to the model of Aparicio, Villaz\'{o}n-Terrazas and \'{A}lvarez (\cite{Aparicio2015AMF}).
\subsection{Tests results of $ \cl{M}_{2p} $ on scale-free DAGs}
The first and perhaps most famous model for the emergence of scale-free graphs, is that of Barab\'asi and Albert. In their model the network starts with an initial set of $ s $ vertices and no edges. At each step we add a new vertex to the graph and connect it to $ k $ existing vertices, which are randomly selected according to distribution which is linearly proportionate to the current degree of the vertices. Thus a vertex with already high degree is more likely to get more new links (a sort of `rich get richer' phenomenon). This model was originally defined to generate undirected graphs. By directing the edges from the new vertex to the old ones, we get a directed model which is also acyclic. Moreover, the initial set of $ s $ vertices are precisely the sinks of the generated DAG. We used this model to simulate academic papers' citation networks, and tested the working of the Two Path mechanism on these networks. We have fixed the number of vertices at 10,000 and simulated networks while changing two parameters: the number of sinks (i.e., size of the initial set), and the out-degree of each new vertex. We observed that as we increase the number of sinks, the approximation ratio of our mechanism worsens. On the other hand, increasing the degree improves the outcome. Notice that this is in accordance to the reasoning of Theorem~\ref{thm: forests}, since a denser DAG with less sinks is more likely to be balanced. Two representative samples of this phenomenon can be seen in Figures~\ref{fig: DAG sinks} and \ref{fig: DAG degree}. In both, each dot represents an average of 100 random networks. For each network we ran the mechanism 100 times and averaged the outcome. Thus, each dot represents 10,000 runs of the mechanism with the appropriate network parameters. Figure~\ref{fig: DAG sinks} shows the deterioration of the ratio when we increase the number of sinks while holding the degree fixed. Figure~\ref{fig: DAG degree} shows the improvement of the ratio when we increase the degree while holding the number of sinks fixed. \\
We also noticed that typically, the approximation ratio was in the range of 3-10. This implies that in real-world networks, the Two Path mechanism can be expected to perform much better than the mathematical bounds we were able to prove.
\setlength{\belowcaptionskip}{0pt}
\setlength{\abovecaptionskip}{0pt}
\setlength{\intextsep}{0pt}

\begin{figure}[h!]
	\centering
	\includegraphics[width={0.85\linewidth}]{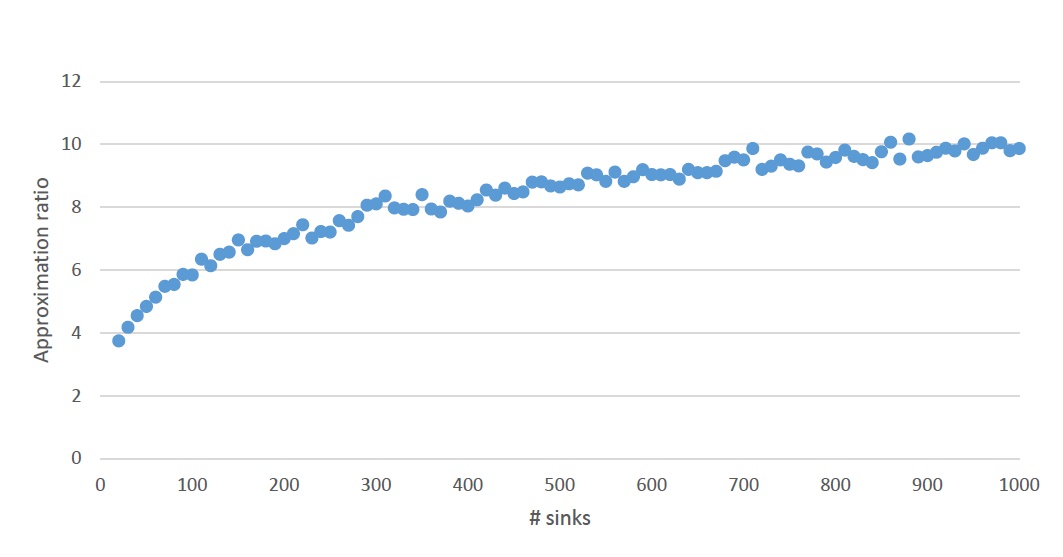}
	\caption{Ratio vs. Sinks. Out-degree = 10.}
	\label{fig: DAG sinks}	
\end{figure}
\setlength{\intextsep}{10pt}
\begin{figure}[h!]
	\centering
	\includegraphics[width={0.85\linewidth}]{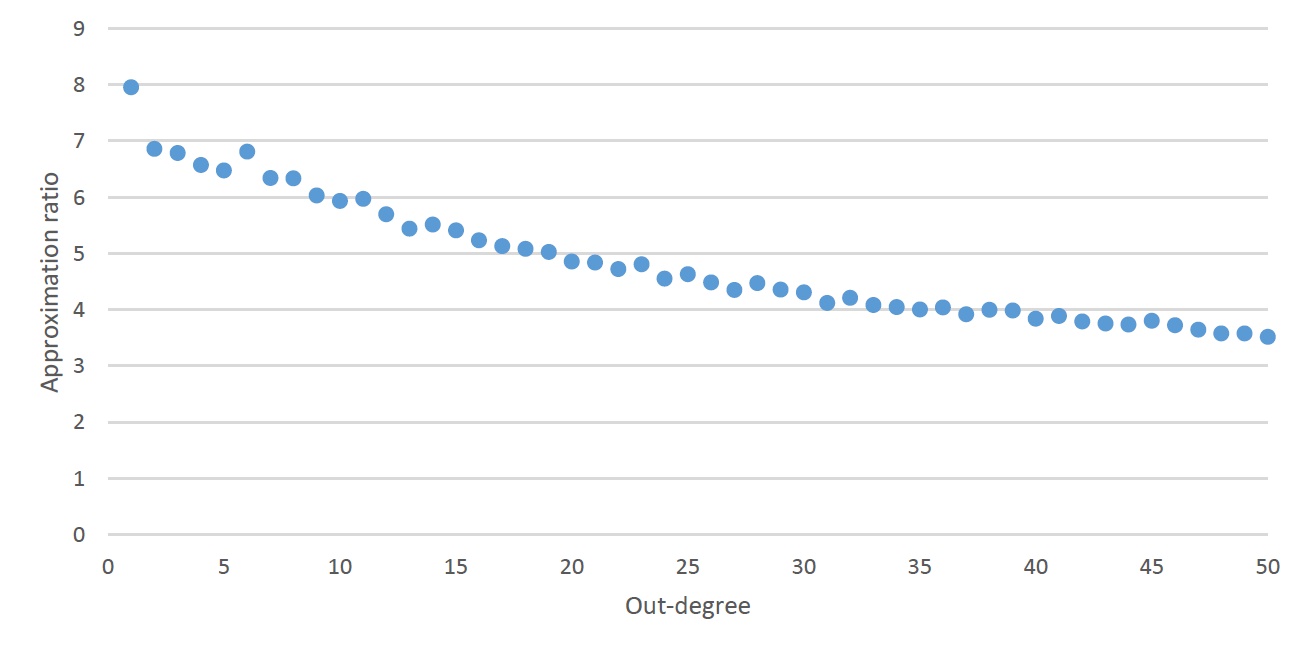}
	\caption{Ratio vs. Out-degree. \# sinks = 100.}
	\label{fig: DAG degree}	
\end{figure}

\subsection{Test results of $ \cl{M}_{2p}^G $ on scale-free networks}

In order to simulate Twitter-like networks, we used the model suggested in \cite{Aparicio2015AMF}. In that paper the authors show that their model is more suited to simulate the dynamic nature of social networks than the Barab\'{a}si-Albert model. In particular, they show that their model generates graphs which resemble Twitter in a few interesting parameters. In their model, at each step one of the following happens.

\begin{itemize}
	\item With probability $ p $ a new vertex is added with an out-going edge to an existing vertex. The target vertex is chosen according to a distribution linearly proportionate to the current in-degree of the vertices.
	\item With probability $ q $ a new vertex is added with an in-going edge from an existing vertex. The source vertex is chosen according to a distribution linearly proportionate to the current out-degree of the vertices.
	\item With probability $ r $ a new edge is added between two existing vertices. The source vertex is chosen according to a distribution linearly proportionate to the current out-degree of the vertices and the target vertex is chosen according to a distribution linearly proportionate to the current in-degree of the vertices.
\end{itemize}
Of course, we need to require that $ p+q+r=1 $. In addition, it is assumed that $ q<p $. We used this model with different parameters to test mechanism $ \cl{M}_{2p}^G $. Again, we fixed the number of vertices at 10,000 and were interested in the influence of two parameters on the performance of our mechanism. The first parameter is the average in-degree in the graph; in the parameters of the model, this is equal to $ 1/(p+q) $\footnote{At each step a new edge is created. The probability of a new vertex in every step is $ p+q $, hence the expected average degree will be $ 1/(p+q) $.}. The second is the probability of a `reverse edge', $\hat{q}=q/(p+q) $, which is the probability of a new vertex to get an edge. We think of this parameter as some indication of how far the graph will be from a DAG (although cycles can also be created on steps where an edge is created between two existing vertices). We observed as before that an increase in the average degree (meaning, an increase in the density of the graph) improves the performance of the mechanism. On the other hand, an increase in $ \hat{q} $ worsens the average outcome. The reason for that might be that some vertices reach high influence due to a few `lucky' edges from other high-influence vertices, while mechanism $ \cl{M}_{2p}^G $ performs better when there is higher correlation between high influence and high in-degree. Again we show two samples of our tests. Figures~\ref{fig: general degree} and \ref{fig: general backedges} both present the running of the mechanism on simulated networks with 10,000 vertices. Each dot represents 100 networks which are tested 100 times each, so a dot is an average of 10,000 outcomes. In the first figure we show the tests in which we held $ \hat{q} $ fixed while increasing the average degree, and in the second figure we show the tests in which we did the opposite.\\
We see that mechanism $ \cl{M}_{2p}^G $ gives a nicely bounded approximation ratio in most scenarios and we thus expect it to perform well in real-life social networks.
\setlength{\belowcaptionskip}{0pt}
\setlength{\abovecaptionskip}{0pt}
\setlength{\intextsep}{0pt}

\begin{figure}[h!]
	\centering
	\includegraphics[width={0.9\linewidth}]{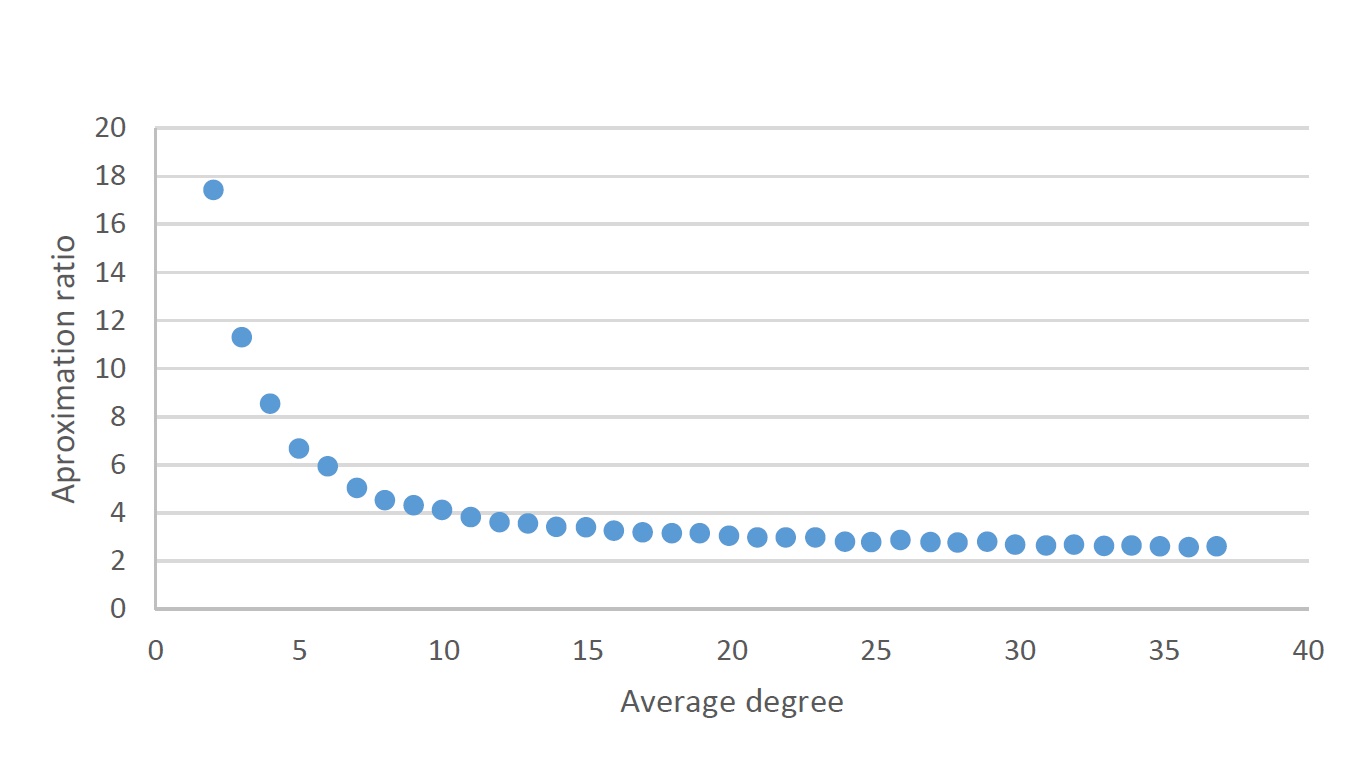}
	\caption{Ratio vs. Average degree. $ \hat{q}=0.15 $.}
	\label{fig: general degree}	
\end{figure}
\setlength{\belowcaptionskip}{-5pt}
\setlength{\intextsep}{0pt}
\begin{figure}[h!]
	\centering
	\includegraphics[width={0.9\linewidth}]{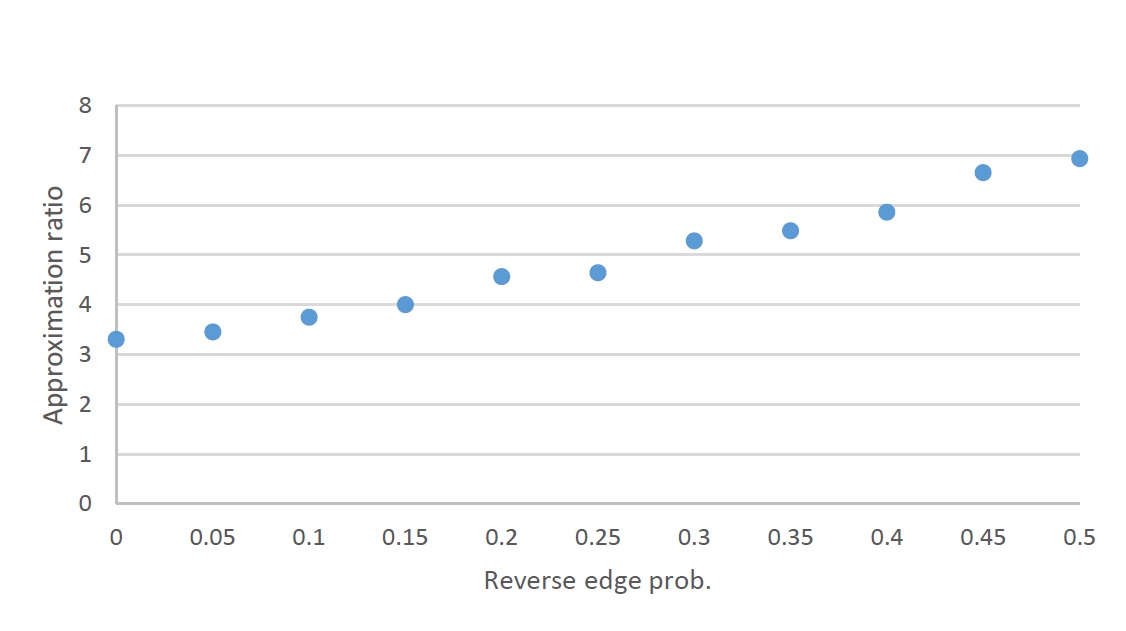}
	\caption{Ratio vs. `Reveres edge' Prob. Average degree = 10.}
	\label{fig: general backedges}	
\end{figure}

\newpage
\bibliographystyle{ACM-Reference-Format}
\balance
\bibliography{paper} 

\end{document}